\documentclass[11pt]{article}
\usepackage{authblk}

\usepackage{amsmath,amsfonts,amssymb,colonequals,mathtools}
\usepackage{amsthm,fullpage}
\newtheorem{theorem}{Theorem}[section]
\newtheorem{lemma}[theorem]{Lemma}

\newtheorem{corollary}[theorem]{Corollary}

\newtheorem{definition}[theorem]{Definition}

\def\ket#1{\left\vert#1\right\rangle}

\renewcommand{\>}{\rangle}

\renewcommand{\(}{\left(}
\renewcommand{\)}{\right)}

\def\braket#1#2{{{\langle}#1\vert}#2\rangle}
\def\abs#1{| #1 |}

\def\spn{\mbox{span}} 
\newcommand{\eps}{\varepsilon}

\newcommand{\zero}{{0}}

\renewcommand{\overrightarrow}[1]{\mathaccent"017E #1} % saves a bit of space
\renewcommand{\emptyset}{\varnothing}
\newcommand{\defeq}{\colonequals}

\def\Garbage{\psi}
\def\colset{\mathcal{P}}
\def\cols{\mathcal{P}}
\def\inds{\mathcal{I}}
\def\inp{\chi}

\begin{document}

\title{A Time-Efficient Quantum Walk for \\ 3-Distinctness Using Nested Updates\thanks{Support for this work was provided by NSERC, the Ontario Ministry of Research and Innovation, the US ARO, the French ANR Blanc project ANR-12-BS02-005 (RDAM), and the European Commission IST STREP project 25596 (QCS).}}
\author[1,3]{Andrew M.~Childs\thanks{amchilds@uwaterloo.ca}}
\author[2,3]{Stacey Jeffery\thanks{sjeffery@uwaterloo.ca}}
\author[2,3]{Robin Kothari\thanks{rkothari@cs.uwaterloo.ca}}
\author[4]{Fr\'ed\'eric Magniez\thanks{frederic.magniez@univ-paris-diderot.fr}}
\affil[1]{\small Department of Combinatorics \& Optimization, University of Waterloo, Canada}
\affil[2]{\small David R.\ Cheriton School of Computer Science, University of Waterloo, Canada}
\affil[3]{Institute for Quantum Computing, University of Waterloo, Canada}
\affil[4]{\small CNRS, LIAFA, Univ Paris Diderot, Sorbonne Paris-Cit\'e, France}

\date{}

\maketitle

\begin{abstract}
We present an extension to the quantum walk search framework that facilitates quantum walks with nested updates.
We apply it to give a quantum walk algorithm 
for $3$-Distinctness with query complexity $\tilde{O}(n^{5/7})$, matching the best known upper bound (obtained via learning graphs) up to log factors. Furthermore, our algorithm has time complexity $\tilde{O}(n^{5/7})$, improving the previous $\tilde{O}(n^{3/4})$.
\end{abstract}

\section{Introduction}
\noindent Element Distinctness is a basic computational problem. 
Given a sequence $\inp=\inp_1,\ldots,\inp_n$ of $n$ integers, the task is to decide if those elements are pairwise distinct.
This problem is closely related to Collision, a fundamental problem in cryptanalysis. 
Given a $2$-to-$1$ function $f:[n]\to [n]$, the aim is to find $a\neq b$ such that $f(a)=f(b)$.
One of the best (classical and quantum) algorithms is to run Element Distinctness on $f$ restricted to a random subset of size
$\sqrt{n}$.

In the quantum setting, Element Distinctness 
has received a lot of attention.
The first non-trivial algorithm used $\tilde{O}(n^{3/4})$ time~\cite{bdh05}. The optimal $\tilde{O}(n^{2/3})$ algorithm is due to
Ambainis~\cite{amb04}, who introduced an approach based on quantum walks that has become a major tool for quantum query algorithms.
The optimality of this algorithm follows from a query lower bound for Collision~\cite{AS04}.
In the query model, access to the input $\inp$ is provided by an oracle whose answer to query $i\in[n]$ is $\inp_i$.
This model is the quantum analog of classical 
decision tree complexity: the only resource measured is the number of queries to the input. 

Quantum query complexity has been a very successful model for studying the power of quantum computation. In particular, quantum query complexity has been exactly characterized in terms of a semidefinite program, the general adversary bound~\cite{rei11,LMR+11}. To design quantum query algorithms, it suffices to exhibit a solution to this semidefinite program. However, this turns out to be difficult in general, as the minimization form of the general adversary bound has exponentially many constraints. 
Belovs~\cite{bel11} recently introduced the model of learning graphs, which can be viewed as the minimization form of the general adversary bound with additional structure imposed on the form of the solution.  This additional structure makes learning graphs much easier to reason about. 
The learning graph model
has already been used to improve the query complexity of many graph problems~\cite{bel11,lms11,LMS13} 
as well as $k$-Distinctness \cite{bel12}.

One shortcoming of 
learning graphs is that these upper bounds 
do not lead explicitly to efficient algorithms in terms of time complexity. 
Although the study of query complexity is interesting on its own, it is relevant in practice only when a query lower bound is close to the best known time complexity. 

Recently, \cite{JKM12} reproduced several known learning graph upper bounds via explicit algorithms in an extension of the quantum walk search framework of \cite{MNRS11}.  
This work produced a new 
quantum algorithmic tool, quantum walks with nested checking. 
Algorithms constructed in the framework of \cite{JKM12} can be interpreted as quantum analogs of randomized algorithms, so they are simple to design and analyze for any notion of cost, including time as well as query complexity. This framework
has interpreted all known learning graphs as quantum walks,
except the very recent \emph{adaptive learning graphs} for $k$-Distinctness~\cite{bel12}.

In $k$-Distinctness, the problem is to decide if there are $k$ copies of the same element in the input,
with $k=2$ being Element Distinctness.
The best lower bound for $k$-Distinctness is the Element Distinctness lower bound $\Omega(n^{2/3})$,
whereas the best query upper bound is $O(n^{1-2^{k-2}/(2^k-1)})=o(n^{3/4})$~\cite{bel12}, achieved using learning graphs, improving the previous bound of $O(n^{k/(k+1)})$ \cite{amb04}. However, the best known time complexity remained $\tilde{O}(n^{k/(k+1)})$. We improve this upper bound for the case when $k=3$.

Our algorithm for $3$-Distinctness is conceptually simple: we walk on sets of $2$-collisions and look for a set containing a $2$-collision that is part of a $3$-collision. We check if a set has this property by searching for an index that evaluates to the same value as one of the $2$-collisions in the set. However, to move to a new set of $2$-collisions, we need to use a quantum walk subroutine for finding $2$-collisions as part of our update step. This simple idea is surprisingly difficult to implement and leads us to develop a new extension of the quantum walk search framework.

Given a Markov chain $P$ with spectral gap $\delta$ and success probability $\eps$ in its stationary distribution,
one can construct a quantum search algorithm 
with cost $\textstyle\mathsf{S}+\frac{1}{\sqrt{\eps}}(\frac{1}{\sqrt{\delta}}\mathsf{U}+\mathsf{C})$~\cite{MNRS11},
where  $\mathsf{S}$, $\mathsf{U}$ and $\mathsf{C}$ are respectively the setup, update, and checking costs
of the quantum analog of $P$.
Using a quantum walk algorithm with costs $\mathsf{S}',\mathsf{U}',\mathsf{C}',\eps',\delta'$ (as in \cite{MNRS11}) as a checking subroutine straightforwardly gives complexity
$\textstyle\mathsf{S}+\frac{1}{\sqrt{\eps}}(\frac{1}{\sqrt{\delta}}\mathsf{U}+\mathsf{S}'+\frac{1}{\sqrt{\eps'}}(\frac{1}{\sqrt{\delta'}}\mathsf{U}'+\mathsf{C}'))$.
Using nested checking \cite{JKM12}, the cost can be reduced to
$\textstyle\mathsf{S}+\mathsf{S}'+\frac{1}{\sqrt{\eps}}(\frac{1}{\sqrt{\delta}}\mathsf{U}+\frac{1}{\sqrt{\eps'}}(\frac{1}{\sqrt{\delta'}}\mathsf{U}'+\mathsf{C}')).$

It is natural to ask if a quantum walk subroutine can be used for the update step in a similar manner
 to obtain cost 
$\textstyle\mathsf{S}+\mathsf{S}'+\frac{1}{\sqrt{\eps}}(\frac{1}{\sqrt{\delta}}\frac{1}{\sqrt{\eps'}}(\frac{1}{\sqrt{\delta'}}\mathsf{U}'+\mathsf{C}')+\mathsf{C}).$
In most applications, the underlying walk is independent of the input, so the update operation is simple, but for some applications a more complex update may be useful (as in \cite{ck11}, where Grover search is used for the update).
In Section~\ref{sec:motiv}, 
we describe an example showing that it is not even clear how to use a nested quantum walk for the update with the seemingly trivial cost 
$\textstyle\mathsf{S}+\frac{1}{\sqrt{\eps}}(\frac{1}{\sqrt{\delta}}(\mathsf{S}'+\frac{1}{\sqrt{\eps'}}(\frac{1}{\sqrt{\delta'}}\mathsf{U}'+\mathsf{C}'))+\mathsf{C})$.
Nevertheless, despite the difficulties that arise in implementing nested updates, we show in Section~\ref{sec:rec} how to achieve the more desirable cost expression in certain cases, and a similar one in general. 

To accomplish this, we extend the quantum walk search framework by introducing the concept of \emph{coin-dependent data}.  This allows us to implement nested updates, with a quantum walk subroutine to carrying out the update procedure. Superficially, our modification appears small.
Indeed, the proof of the complexity of our framework is nearly the same as that of \cite{MNRS11}.
However, there are some subtle differences in the implementation of the walk.

As in \cite{JKM12}, this concept is simple yet powerful.  We demonstrate this by constructing a quantum walk version of the learning graph for $3$-Distinctness with matching query complexity (up to poly-logarithmic factors). Because quantum walks are easy to analyze, the time complexity, which matches the query complexity, follows easily, answering an open problem of~\cite{bel12}. 

Independently, Belovs \cite{bel13} also recently obtained a time-efficient implementation of his learning graph for $3$-Distinctness. His approach also uses quantum walks, but beyond this similarity, the algorithm appears quite different.  In particular, it is based on another framework of search via quantum walk due to Szegedy~\cite{sze04,mnrs12}, whereas our approach uses a new extension of the quantum walk search framework of \cite{MNRS11}.

\section{Preliminaries and Motivation}

\subsection{Quantum Walks}\label{app:qw}

Consider a reversible, ergodic Markov chain $P$ on a connected, undirected graph $G=(X,E)$
with spectral gap $\delta>0$ and stationary distribution $\pi$. 
Let $M\subseteq X$ be a set of marked vertices. Our goal is to detect wether $M=\emptyset$ or $\Pr_{x\sim\pi} (x\in M) \geq \varepsilon$,
for some given $\varepsilon>0$.
Consider the following randomized algorithm that finds a marked element with bounded error.
\begin{quote}
\begin{enumerate}\item Sample~$x$ from $\pi$
\item Repeat for~$\Theta(1/\eps)$ steps
\begin{enumerate}
\item\label{item-check} If the current vertex~$x$  is marked,
  then stop and output~$x$
\item Otherwise, simulate~$\Theta(1/\delta)$ steps of $P$ starting with $x$
\end{enumerate}
\item If the algorithm has not terminated,  output `no marked
  element'
\end{enumerate}
\end{quote}

This algorithm has been quantized by~\cite{MNRS11} leading to efficient quantum query algorithms. 
Since each step has to be unitary and therefore reversible, we have to implement the walk carefully.
The quantization considers $P$ as a walk on edges of $E$. 
We write $(x,y)\in \overrightarrow{E}$ when we consider an edge $\{x,y\}\in E$ with orientation $(x,y)$.
The notation $(x,y)$ intuitively means that the current vertex of the walk is $x$ and the \emph{coin}, indicating the next move, is $y$. Swapping $x$ and $y$ changes the current vertex to $y$; then the coin becomes $x$.

The quantum algorithm may carry some data structure while walking on $G$; we formalize this as follows.
Let $\zero$ be a state outside $X$.
Define ${D}:X\cup\{0\}\rightarrow \mathcal{D}$ for some Hilbert space $\mathcal{D}$, with $\ket{{D}(0)}=\ket{0}$. 
We define costs associated with the main steps of the algorithm. By cost we mean any measure of complexity such as query, time or space.

\begin{description}
\item[\emph{Setup cost}:] Let $\mathsf{S}$ be the cost of constructing
\[
\ket{\pi}=\sum_{x\in X}\sqrt{\pi(x)}\ket{x}\ket{D(x)}\sum_{y\in X}\sqrt{P(x,y)}\ket{y}\ket{{D}(y)}.
\]
\item[\emph{Update cost}:] Let $\mathsf{U}$ be the cost of the \textsc{Local Diffusion} operation, which is controlled on the first two registers\footnote{The requirement that this operation be controlled on the first two registers, i.e., that it always leaves the first two registers unchanged, is not explicitly stated in \cite{MNRS11}. However, this condition is needed to prevent, for example, the action $\ket{x}\ket{\psi}\ket{0,0}\mapsto \ket{x}\ket{D(x)}\ket{\phi}$, where $\braket{\psi}{D(x)}=0$ and $\ket{x}\ket{D(x)}\ket{\phi}$ is a possible state of the algorithm. In this case, $(\textsc{Local Diffusion})\text{ref}_{\ket{0,0}}(\textsc{Local Diffusion})^\dagger$ would not act as $W(P)$ on $\ket{x}\ket{D(x)}\ket{\phi}$.} and acts as
\[
\ket{x}\ket{{D}(x)}\ket{0}\ket{D(0)}\mapsto \ket{x}\ket{D(x)}\sum_{y\in X}\sqrt{P(x,y)}\ket{y}\ket{{D}(y)}.
\]
\item[\emph{Checking cost}:] Let $\mathsf{C}$ be the cost of the reflection
\[
\ket{x}\ket{{D}(x)}\mapsto \begin{cases}
-\ket{x}\ket{{D}(x)} & \mbox{if }x\in M\\
\ket{x}\ket{{D}(x)} & \mbox{otherwise.}\end{cases}
\]
\end{description}

\begin{theorem}[\cite{MNRS11}]
Let $P$ be a reversible, ergodic Markov chain on $G=(X,E)$ with spectral gap $\delta>0$. 
Let $M\subseteq X$ be such that $\Pr_{x\sim \pi}(x\in M)\geq \eps$, for some $\eps>0$, whenever $M\neq\emptyset$. 
Then there is a quantum algorithm that finds an element of $M$, if $M \ne \emptyset$,
with bounded error and with cost
$$\textstyle O\left(\mathsf{S}+\frac{1}{\sqrt{\eps}}\left(\frac{1}{\sqrt{\delta}}\mathsf{U}+\mathsf{C}\right)\right).$$
\end{theorem}
Furthermore, we can approximately map $\ket{\pi}$ to $\ket{\pi(M)}$, the normalized projection of $\ket{\pi}$ onto $\spn\{\ket{x}\ket{D(x)}\ket{y}\ket{D(y)}:x\in M,y\in X\}$, in cost $\frac{1}{\sqrt{\eps}}(\frac{1}{\sqrt{\delta}}\mathsf{U}+\mathsf{C})$.

\subsection{$3$-Distinctness}\label{sec:prelim-3dist}
We suppose that the input is a sequence $\inp=\inp_1,\ldots,\inp_n$ of integers from $[q] \defeq \{1,\ldots,q\}$.
We model the input as an oracle whose answer to query $i\in[n]$ is $\inp_i$.

We make the simplifying assumptions that there is at most one $3$-collision and that the number of $2$-collisions is in $\Theta(n)$. The first assumption is justified in \cite[Section 5]{amb04}. To justify the second assumption, note that given an input $\inp\in [q]^n$, we can construct $\inp'\in [q+n]^{3n}$ with the same $3$-collisions as $\inp$, 
and $\Omega(n)$ $2$-collisions, by defining $\inp'_i=\inp_i$ for $i\in [n]$ and $\inp_i'=\inp_{i+n}'=q+i$ for $i\in \{n+1,\dots,2n\}$. Note that any two $2$-collisions not both part of the $3$-collision are disjoint.

A common simplifying technique is to randomly partition the space $[n]$ and assume that the solution respects the partition in some sense. Here we partition the space into three disjoint sets of equal size, $A_1$, $A_2$ and $A_3$, 
and assume that if there is a $3$-collision $\{i,j,k\}$, then we have $i\in A_1$, $j \in A_2$ and $k\in A_3$. This assumption holds with constant probability, so we need only repeat the algorithm $O(1)$ times with independent choices of the tripartition to find any $3$-collision with high probability.
Thus, we assume we have such a partition.

\subsection{Motivating Example}\label{sec:motiv}

\paragraph{Quantum Walk for Element Distinctness.} In the groundbreaking work of Ambainis 
\cite{amb04}, which inspired a series of quantum walk frameworks \cite{sze04,MNRS11,JKM12} leading up to this work, 
a quantum walk for solving Element Distinctness was presented. This walk takes place on a Johnson graph, $J(n,r)$, whose vertices are subsets of $[n]$ of size $r$, denoted $\binom{[n]}{r}$. In $J(n,r)$, two vertices $S,S'$ are adjacent if $\abs{S\cap S'}=r-1$. 
The data function is ${D}(S)=\{(i,\inp_i):i\in S\}$. 
The diffusion step of this walk acts as
$$\textstyle\ket{S}\ket{{D}(S)}\ket{0}\mapsto \ket{S}\ket{D(S)}\frac{1}{\sqrt{r(n-r)}}\sum_{i\in S,j\in [n]\setminus S}\ket{(S\setminus i)\cup j}\ket{{D}((S\setminus i)\cup j)}.$$
\noindent We can perform this diffusion in two queries by performing the transformation
$$\textstyle\ket{S}\ket{{D}(S)}\ket{0}\mapsto \ket{S}\ket{{D}(S)}\frac{1}{\sqrt{r}}\sum_{i\in S}\ket{(i,\inp_i)}\frac{1}{\sqrt{n-r}}\sum_{j\in [n]\setminus S}\ket{(j,\inp_j)}.$$
\noindent We can reversibly map this to the desired state with no queries, and by using an appropriate encoding of ${D}$, we can make this time efficient as well. 

To complete the description of this algorithm, we describe the marked set and checking procedure. We deviate slightly from the usual quantum walk algorithm of \cite{amb04} and instead describe a variation that is 
analogous to the learning graph for Element Distinctness \cite{bel11}. We say a vertex $S$ is marked if it contains an index $i$ such that there exists $j\in[n]\setminus\{i\}$ with $\inp_i=\inp_j$ (whereas in \cite{amb04} both $i$ and $j$ must be in $S$). To check if $S$ is marked, we simply search over $[n]\setminus S$ for such a $j$, in cost $O(\sqrt{n})$. This does not give asymptotically better performance 
than \cite{amb04}, but it is more analogous to the $3$-Distinctness algorithm we  attempt to construct in the remainder of this section, and then succeed in constructing in Section \ref{sec:query}.

\paragraph{Attempting a Quantum Walk for $3$-Distinctness.} We now attempt to construct an analogous algorithm for $3$-Distinctness. 
Conceptually, the approach is simple, but successfully implementing the simple idea is nontrivial. 
The idea is to walk on a Johnson graph of sets of \emph{collision pairs}, analogous to the set of queried indices in the Element Distinctness walk described above.  The checking step is then similar to that of the above walk: simply search for a third element that forms a $3$-collision with one of the $2$-collisions in the set. For the update step, we need to replace one of the collision pairs in the set using a subroutine that finds a $2$-collision. We now describe the difficulty of implementing this step efficiently, despite having an optimal Element Distinctness algorithm at our disposal. Section \ref{sec:qw-nested-coin} presents a framework that allows us to successfully implement the idea in Section \ref{sec:query}.

Let $\colset$ denote the set of collision pairs in the input, and $n_2=\abs{\colset}$. We walk on $J(n_2,s_2)$, with each vertex $S_2$ corresponding to a set of $s_2$ collision pairs. 
The diffusion for this walk is the map
$$\ket{S_2,{D}(S_2)}\ket{0}\mapsto
\textstyle\ket{S_2,D(S_2)}\tfrac{1}{\sqrt{r(n_2-s_2)}}\sum_{\substack{(i,i')\in S_2\\(j,j')\in \colset\setminus S_2}}\ket{(S_2\setminus (i,i'))\cup (j,j')}\ket{{D}((S_2\setminus (i,i'))\cup (j,j'))}. 
$$
\noindent To accomplish this, we need to generate $\frac{1}{\sqrt{s_2}}\sum_{(i,i')\in S_2}\ket{(i,i',\inp_i)}$ and $\frac{1}{\sqrt{n_2-s_2}}\sum_{(j,j')\in\colset\setminus S_2}\ket{(j,j',\inp_j)}$. The first superposition is easy to generate, since we have $S_2$, but the second is more difficult since we have to find new collisions.

The obvious approach is to use the quantum walk algorithm for Element Distinctness as a subroutine. However, this algorithm does not return the desired superposition over collisions; rather, it returns 
a superposition over sets that contain a collision.
That is, we have the state 
$\frac{1}{\sqrt{n_2}}\sum_{(i,i')\in\colset}\ket{(i,i',\inp_i)}\ket{\Garbage(i,i')}$ for some garbage $\ket{\Garbage(i,i')}$.
The garbage may be only slightly entangled with $(i,i')$, 
but even this small amount of error in the state is prohibitive.  Since we must call the update subroutine many times, we need the error to be very small.  Unlike for nested checking, where bounded-error subroutines are sufficient, we cannot amplify the success probability of an update operator.
We cannot directly use the state returned by the Element Distinctness algorithm for several reasons. First, we cannot append garbage each time we update, as
this would prevent proper interference in the walk. Second, when we use a nested walk for the update step, we would like to use the same trick as in nested checking: putting a copy of the starting state for the nested walk in the data structure so that we only need to perform the inner setup once. To do the same here, we would need to preserve the inner walk starting state; in other words, the update would need to output some state close to
$\binom{n}{s_1}^{-1/2}\sum_{S_1\in\binom{[n]}{s_1}}\ket{S_1}$.
While we might try to recycle the garbage to produce this state,
it is unclear how to extract the part we need for the update coherently, let alone without damaging the rest of the state. 

This appears to be a problem for any approach that directly uses a quantum walk for the update, since all known quantum walks use some variant of a Johnson graph.
Our modified framework circumvents this issue by allowing us to do the update with some garbage, which we then uncompute. This lets us use a quantum walk subroutine, with setup performed only at the beginning of the algorithm, to accomplish the update step. More generally, using our modified framework, we can tolerate updates that have garbage for any reason, whether the garbage is the result of the update being implemented by a quantum walk, or by some other quantum subroutine.

\section{Quantum Walks with Nested Updates}\label{sec:qw-nested-coin}
\subsection{Coin-Dependent Data}\label{sec-coin}

A quantum analog of a discrete-time random walk on a graph can be constructed as a unitary process on the directed edges. For an edge $\{x,y\}$, we may have a state $\ket{x}\ket{y}$, where $\ket{x}$ represents the current vertex and $\ket{y}$ represents the \emph{coin} or next vertex. In the framework of \cite{MNRS11}, some data function on the vertices is employed to help implement the search algorithm. 
We modify the quantum walk framework to allow this data 
to depend on both the current vertex and the coin, so that it is a function of the directed edges, which seems natural in hindsight. We show that this point of view has algorithmic applications.
In particular, this modification enables efficient nested updates. 

In the rest of the paper, let $P$ be a reversible, ergodic
Markov chain on a connected, undirected graph $G=(X,E)$
with stationary distribution $\pi$ and
spectral gap $\delta>0$.

 Let $\zero\not\in X$. Let ${D}:(X\times \{\zero\}) \cup \overrightarrow{E}\rightarrow \mathcal{D}$ 
for some Hilbert space $\mathcal{D}$. 
A quantum analog of $P$ with coin-dependent data structures can be implemented using three operations, as in \cite{MNRS11}, but the update now has three parts.
The first corresponds to \textsc{Local Diffusion} from the framework of~\cite{MNRS11}, as described in
Section~\ref{app:qw}. The others are needed because of the new coin-dependent data.

\begin{description}
\item[\emph{Update cost}:] Let $\mathsf{U}$ be the cost of implementing
\begin{itemize}
\item \textsc{Local Diffusion}:
$\ket{x,0}\ket{{D}(x,0)}\mapsto \sum_{y\in X}\sqrt{P(x,y)}\ket{x,y}\ket{{D}(x,y)}$
$\forall\, x\in X$;
\item The $(X,0)$-\textsc{Phase Flip}:
$\ket{x,0}\ket{{D}(x,0)}\mapsto - \ket{x,0}\ket{{D}(x,0)}$ $\forall\, x \in X$, and the identity on the orthogonal subspace; and
\item The \textsc{Database Swap}:
$\ket{x,y}\ket{{D}(x,y)}\mapsto \ket{y,x}\ket{{D}(y,x)}$
$\forall\, (x,y)\in \overrightarrow{E}$.
\end{itemize}
\end{description}

\noindent By cost, we mean any desired measure of complexity such as queries, time, or space. We also naturally extend the setup and checking costs as follows, where $M\subseteq X$ is a set of marked vertices.

\begin{description}
\item[\emph{Setup cost}:] Let $\mathsf{S}$ be the cost of constructing
\[\ket{\pi}\defeq\sum_{x\in X}\sqrt{\pi(x)}\sum_{y\in X}\sqrt{P(x,y)}\ket{x,y}\ket{{D}(x,y)}.
\]
\item[\emph{Checking cost}:] Let $\mathsf{C}$ be the cost of the reflection
\[\ket{x,y}\ket{{D}(x,y)}\mapsto\begin{cases}
-\ket{x,y}\ket{{D}(x,y)} & \mbox{if }x\in M,\\
\ket{x,y}\ket{{D}(x,y)} & \mbox{otherwise,}\end{cases}
~\forall\, (x,y)\in\overrightarrow{E}.
\]
\end{description}

Observe that $\ket{\pi}^0\defeq\sum_{x\in X} \sqrt{\pi(x)}\ket{x,0}\ket{D(x,0)}$ can be mapped to $\ket{\pi}$ by the \textsc{Local Diffusion}, which has cost $\mathsf{U}<\mathsf{S}$, so we can also consider $\mathsf{S}$ to be the cost of constructing $\ket{\pi}^0$.

\begin{theorem}\label{thm:coin}
Let $P$ be a 
 Markov chain on $G=(X,E)$ 
with spectral gap $\delta>0$, and let $D$ be a \emph{coin-dependent} data structure for $P$.
Let $M\subseteq X$ satisfy $\Pr_{x\sim \pi}(x\in M)\geq \eps > 0$ whenever $M\neq\emptyset$. 
Then there is a quantum algorithm that finds an element of $M$, if $M \ne \emptyset$,
with bounded error and with cost
$$\textstyle O\left(\mathsf{S}+\frac{1}{\sqrt{\eps}}\left(\frac{1}{\sqrt{\delta}}\mathsf{U}+\mathsf{C}\right)\right).$$
\end{theorem}
\begin{proof}
Our quantum walk algorithm is nearly identical to that of \cite{MNRS11}, so the proof of this theorem is also very similar. Just as in \cite{MNRS11}, we define a walk operator, $W(P)$, and analyze its spectral properties. 
Let $\mathcal{A}\defeq\spn\{\sum_y\sqrt{P(x,y)}\ket{x,y}\ket{D(x,y)}:x\in X\}$ and define
$W(P)\defeq((\textsc{Database Swap})\cdot \mathrm{ref}_\mathcal{A})^2$, where $\mathrm{ref}_\mathcal{A}$ denotes the reflection about $\mathcal{A}$.

As in \cite{MNRS11}, we can define 
$\mathcal{H}\defeq\spn\{\ket{x,y}: (x,y)\in (X\times \{\zero\})\cup\overrightarrow{E}\}$
and 
$\mathcal{H}_D\defeq\spn\{\ket{x,y,{D}(x,y)}: (x,y)\in (X\times \{\zero\})\cup\overrightarrow{E}\}$.
As in \cite{MNRS11}, there is a natural isomorphism $\ket{x,y}\mapsto\ket{x,y}_D=\ket{x,y,{D}(x,y)}$, and $\mathcal{H}_D$ is invariant under both $W(P)$ and the checking operation. Thus, the spectral analysis may be done in $\mathcal{H}$, on states without data, \emph{exactly} as in \cite{MNRS11}. 
However, there are some slight differences in how we implement $W(P)$, which we now discuss.

The first difference is easy to see: in \cite{MNRS11}, the \textsc{Database Swap} can be accomplished trivially by a $\text{SWAP}$ operation, mapping $\ket{x}\ket{y}\ket{D(x)}\ket{D(y)}$ to $\ket{y}\ket{x}\ket{D(y)}\ket{D(x)}$, whereas in our case, 
there may be a nontrivial cost associated with the mapping $\ket{D(x,y)}\mapsto \ket{D(y,x)}$, which we must include in the calculation of the update cost.

The second difference is more subtle. 
In \cite{MNRS11}, $\mathrm{ref}_\mathcal{A}$ is implemented by applying $\textsc{(Local Diffusion)}^\dagger$, reflecting about $\ket{0,D(0)}$
(since the data only refers to a vertex)
 in the coin register, and then applying $\textsc{(Local Diffusion)}$. 
 It is simple to reflect about $\ket{0,D(0)}$, since $\ket{D(0)}=\ket{0}$ in the formalism of~\cite{MNRS11}. In \cite{MNRS11}, this reflection is sufficient, because the operation $(\textsc{Local Diffusion})^\dagger$ fixes the vertex and its data, $\ket{x}\ket{D(x)}$, so in particular, it is still in the space $\spn\{\ket{x}\ket{D(x)}:x\in X\}$.  The register containing the coin and its data, $\ket{y}\ket{D(y)}$, may be moved out of this space by $(\textsc{Local Diffusion})^\dagger$, so we must reflect about $\ket{0}\ket{D(0)}$, but this is straightforward.

With coin-dependent data, a single register $\ket{D(x,0)}$ holds the data for both the vertex and its coin, and the operation $(\textsc{Local Diffusion})^\dagger$ may take the coin as well as the entire data register out of the space $\mathcal{H}_D$, so we need to reflect about $\ket{0}\ket{{D}(x,0)}$, which is not necessarily defined to be $\ket{0}\ket{0}$. 
This explains why the cost of $(X,0)$-\textsc{Phase Flip} is also part of the update cost. In summary, we implement $W(P)$ by 
$(\textsc{(Database Swap)}\cdot \textsc{(Local Diffusion)}\cdot($(X,0)$-\textsc{Phase Flip})\cdot\textsc{(Local Diffusion)}^\dagger)^2$.
\end{proof}

\subsection{Nested Updates}\label{sec:rec}

We show how to implement efficient nested updates using the coin-dependent data framework. 
 Let ${C}: X\cup\{0\}\rightarrow \mathcal{C}$ be some coin-independent data structure (that will be a part of the final data structure) with $\ket{C(0)}=\ket{0}$,
where we can 
reflect about $\spn\{\ket{x}\ket{C(x)}:x\in M\}$ in cost $\mathsf{C}_C$. 
 In the motivating example, if $x=S_2$ is a set of collision pairs, then $C(S_2)$ stores their query values.

Fix $x\in X$. Let $P^{x}$ be a walk on a graph $G^x=(V^x,E^x)$ 
with stationary distribution $\pi^{x}$ and marked set $M^{x}\subset V^{x}$. 
We use this walk to perform \textsc{Local Diffusion} over $\ket{x}$. 
Let ${d}^{x}$ be the 
data for this walk.

When there is ambiguity, we specify the data structure with a subscript.
For instance,
$\ket{\pi}_D=\sum_{x,y\in X} \sqrt{\pi(x)P(x,y)}\ket{x,y}\ket{D(x,y)}$ and
$\ket{\pi}^0_C=\sum_{x\in X}\sqrt{\pi(x)}\ket{x,0}\ket{C(x),0}$. 
Similarly, $\mathsf{S}_C$ is the cost to construct the state $\ket{\pi}_C$.

\begin{definition}
The family $(P^x,M^x,d^x)_{x\in X}$ {\em implements the \textsc{Local Diffusion} and \textsc{Database Swap} of $(P,C)$ with cost $\mathsf{T}$}
if the following two maps can be implemented with cost $\mathsf{T}$:\\[10pt]
\textsc{Local Diffusion with Garbage}: For some garbage states $(\ket{\Garbage(x,y)})_{(x,y)\in\overrightarrow{E}}$, an operation controlled on the vertex $x$ and $C(x)$, acting as
\[
  \ket{x,0}\ket{C(x),0}\ket{\pi^x(M^x)}_{d^x}
\mapsto\sum_{y\in X}\sqrt{P(x,y)}\ket{x,y}\ket{C(x),C(y)}\ket{\Garbage(x,y)};
\]
\textsc{Garbage Swap}: For any edge $(x,y)\in \overrightarrow{E}$,
\[ \ket{x,y}\ket{C(x),C(y)}\ket{\Garbage(x,y)}\mapsto\ket{y,x}\ket{C(y),C(x)}\ket{\Garbage(y,x)}.
\]
The {\em data structure of the implementation} is
$\ket{{D}(x,0)}=\ket{C(x),0}\ket{\pi^x(M^x)}_{d^x}$ for all $x\in X$
and 
$\ket{{D}(x,y)}=\ket{C(x),C(y)}\ket{\Garbage(x,y)}$ for any edge $(x,y)\in \overrightarrow{E}$.
\end{definition}

\begin{theorem}\label{th:nested-update}
Let $P$ be a reversible, ergodic Markov chain on $G=(X,E)$ with 
spectral gap $\delta>0$, and let $C$ be a data structure for $P$.
Let $M\subseteq X$ be such that $\Pr_{x\sim \pi}(x\in M)\geq \eps$ for some $\eps>0$ whenever $M\neq\emptyset$.
Let $(P^x,M^x,d^x)_{x\in X}$ be a family implementing the \textsc{Local Diffusion} and \textsc{Database Swap} of $(P,C)$ with cost $\mathsf{T}$,
and let $\mathsf{S}',\mathsf{U}',\mathsf{C}',1/\eps',1/\delta'$ be upper bounds 
on the costs and parameters associated with each of the $(P^{x},M^x,d^x)$. 
Then there is a quantum algorithm that finds an element of $M$, if $M \ne \emptyset$,
with bounded error and with cost
$$\textstyle\tilde{O}\left(\mathsf{S}_C+\mathsf{S}'+\frac{1}{\sqrt{\eps}}\left(\frac{1}{\sqrt{\delta}}\left(\frac{1}{\sqrt{\eps'}}\left(\frac{1}{\sqrt{\delta'}}\mathsf{U}'+\mathsf{C}'\right)+\mathsf{T}\right)+\mathsf{C}_C\right)\right).$$
\end{theorem}
\begin{proof}
We achieve this upper bound using the quantization of $P$ with the data structure of the implementation, ${D}$. We must compute the cost of the setup, update, and checking operations associated with this walk. 

\begin{description}

\item[\emph{Checking}:] The checking cost $\mathsf{C}=\mathsf{C}_D$
 is the cost to reflect about $\spn\{\ket{x}\ket{y}\ket{D(x,y)}:x\in M\}=\spn\{\ket{x}\ket{y}\ket{C(x),C(y)}\ket{\psi(x,y)}:x\in M\}$. We can implement this in $\mathcal{H}_D$ by reflecting about $\spn\{\ket{x}\ket{C(x)}:x\in M\}$, which costs $\mathsf{C}_C$. 

\item[\emph{Setup}:] 
Recall that $\ket{C(0)}=\ket{0}$.
The setup cost $\mathsf{S}=\mathsf{S}_D$ is the cost of constructing the state 
$$\textstyle\sum_{x\in X}\sqrt{\pi(x)}\ket{x}\ket{0}\ket{{D}(x,0)}=\sum_{x\in X}\sqrt{\pi(x)}\ket{x}\ket{0}\ket{C(x),0}\ket{\pi^x(M^x)}.$$
\noindent We do this as follows. We first construct $\sum_{x\in X}\sqrt{\pi(x)}\ket{x,0}\ket{C(x),0}$ in cost $\mathsf{S}_C$. Next, we apply the mapping $\ket{x}\mapsto \ket{x}\ket{\pi^x}$ in cost $\mathsf{S}'$. Finally, we use the quantization of $P^x$ to perform the mapping $\ket{x}\ket{\pi^x}\mapsto \ket{x}\ket{\pi^x(M^x)}$ in cost $\frac{1}{\sqrt{\eps'}}(\frac{1}{\sqrt{\delta'}}\mathsf{U}'+\mathsf{C}')$. The full setup cost is then
$\mathsf{S}=\mathsf{S}_C+\mathsf{S}'+\frac{1}{\sqrt{\eps'}}(\frac{1}{\sqrt{\delta'}}\mathsf{U}'+\mathsf{C}').$

\item[\emph{Update}:] The update cost 
has three contributions. The first is the \textsc{Local Diffusion} operation, which, by the definition of ${D}$, is exactly the \textsc{Local Diffusion with Garbage} operation.  Similarly, the \textsc{Database Swap} is exactly the \textsc{Garbage Swap}, so these two operations have total cost $\mathsf{T}$. The $(X,0)$-\textsc{Phase Flip} is simply a reflection about states of the form $\ket{x}\ket{{D}(x,0)}=\ket{x}\ket{C(x)}\ket{\pi^x(M^x)}$. Given any $x\in X$, we can reflect about $\ket{\pi^x(M^x)}$ using the quantization of $P^x$ in cost $\frac{1}{\sqrt{\eps'}}(\frac{1}{\sqrt{\delta'}}\mathsf{U}'+\mathsf{C}')$ by running the algorithm of Theorem~\ref{thm:coin}.  In particular, we can run the walk backward to prepare the state $\ket{\pi^x}$, perform phase estimation on the walk operator to implement the reflection about this state, and then run the walk forward to recover $\ket{\pi^x(M^x)}$.
However, this transformation is implemented approximately. To keep the overall error small, we need an accuracy of $O(1/\sqrt{\eps\delta\eps'\delta'})$, which leads to
an overhead logarithmic in the required accuracy. 
The reflection about $\ket{\pi^x(M^x)}$, controlled on $\ket{x}$, is sufficient because \textsc{Local Diffusion with Garbage} is controlled on $\ket{x}\ket{C(x)}$, and so it leaves these registers unchanged. Since we apply the $(X,0)$-\textsc{Phase Flip} just after applying $(\textsc{Local Diffusion})^\dagger$ (see proof of Theorem \ref{thm:coin}) to a state in $\mathcal{H}_D$, we can guarantee that these registers contain $\ket{x}\ket{C(x)}$ for some $x\in X$.
 The total update cost (up to log factors) 
is $\mathsf{U}=\mathsf{T}+\frac{1}{\sqrt{\eps'}}(\frac{1}{\sqrt{\delta'}}\mathsf{U}'+\mathsf{C}').$
\end{description}

Finally, the full cost of the quantization of $P$ (up to log factors) is
\begin{align*}
&\textstyle\mathsf{S}_C+\mathsf{S}'+\frac{1}{\sqrt{\eps'}}\(\frac{1}{\sqrt{\delta'}}\mathsf{U}'+\mathsf{C}'\)+\frac{1}{\sqrt{\eps}}\(\frac{1}{\sqrt{\delta}}\(\frac{1}{\sqrt{\eps'}}\(\frac{1}{\sqrt{\delta'}}\mathsf{U}'+\mathsf{C}'\)+\mathsf{T}\)+\mathsf{C}_C\) \\
&\quad\textstyle=\tilde{O}\(\mathsf{S}_C+\mathsf{S}'+\frac{1}{\sqrt{\eps}}\(\frac{1}{\sqrt{\delta}}\(\frac{1}{\sqrt{\eps'}}\(\frac{1}{\sqrt{\delta'}}\mathsf{U}'+\mathsf{C}'\)+\mathsf{T}\)+\mathsf{C}_C\)\). \qedhere
\end{align*}
\end{proof}

If $\mathsf{T}=0$ (as may be the case, e.g., when the notion of cost  is query complexity), then the expression is exactly what we would have liked for nested updates.

%%%%%%%%%%%%%%%%%%%%%%%%%%%%%%%%%%%%%%%%%%%%%%%%%%%%%%%%%%%%%%%%%%%%%%%%%%%%%%%%%

\section{Application: Quantum Query Complexity of 3-Distinctness}\label{sec:query}

In this section we prove the following theorem.

\begin{theorem}\label{th:query-main}
The quantum query complexity of $3$-Distinctness is $\tilde{O}(n^{5/7})$.
\end{theorem}

We begin by giving a high-level description of the quantum walk algorithm before describing the implementation of each required procedure and their costs. First we define some notation. 

For any set $S_1\subseteq A_1 \cup A_2$, let $\cols(S_1)\defeq\{(i,j) \in A_1 \times A_2: i,j  \in S_1, i\neq j,\inp_i=\inp_j\}$ be the set of $2$-collisions in $S_1$ and for any set $S_2\subset A_1\times A_2$, let $\inds(S_2)\defeq \bigcup_{(i,j)\in S_2}\{i,j\}$ be the set of indices that are part of pairs in $S_2$. In general, we only consider $2$-collisions in $A_1\times A_2$; other $2$-collisions in $\inp$ are ignored. For any pair of sets $A,B$, let $\colset(A,B) \defeq \{(i,j) \in A \times B:i\neq j,\inp_i=\inp_j\}$ be the set of $2$-collisions between $A$ and $B$. For convenience, we define $\colset \defeq \colset(A_1, A_2)$
Let $n_2 \defeq \abs{\colset}$ be the size of this set. For any set $S_2\subseteq\colset$, we denote the set of queried values by $Q(S_2)\defeq\{(i,j,\inp_i):(i,j)\in S_2\}$. Similarly, for any set $S_1\subset [n]$, we denote the set of queried values by $Q(S_1)\defeq\{(i,\inp_i):i\in S_1\}.$

\subsection{High-Level Description of the Walk}

\paragraph{The Walk} Our overall strategy is to find a $2$-collision $(i,j) \in A_1 \times A_2$ such that $\exists k\in A_3$ with $\{i,j,k\}$ a $3$-collision. Let $s_1,s_2<n$ be parameters to be optimized. We walk on the vertices $X=\binom{\colset}{s_2}$, with each vertex corresponding to a set of $s_2$ $2$-collisions from $A_1 \times A_2$. A vertex is considered marked if it contains $(i,j)$ such that $\exists k\in A_3$ with $\{i,j,k\}$ a $3$-collision. Thus, if $M\neq \emptyset$, the proportion of marked vertices is $\eps = \Omega(\frac{s_2}{n_2})$. 

To perform an update, we use an Element Distinctness subroutine that walks on $s_1$-sized subsets of $A_1 \cup A_2$. However, since $n_2$ is large by assumption, the expected number of collisions in a set of size $s_1$ is large if $s_1 \gg \sqrt{n}$, which we suppose holds. It would be a waste to take only one and leave the rest, so we replace multiple elements of $S_2$ in each step. This motivates using a generalized Johnson graph $J(n_2,s_2,m)$ for the main walk, where we set $m \defeq \frac{s_1^2n_2}{n^2}=O(\frac{s_1^2}{n})$, the expected number of $2$-collisions in a set of size $s_1$. In $J(n_2,s_2,m)$, two vertices $S_2$ and $S_2'$ are adjacent if $\abs{S_2\cap S_2'}=s_2-m$, so we can move from $S_2$ to $S_2'$ by replacing $m$ elements of $S_2$ by $m$ distinct elements. Let $\Gamma(S_2)$ denote the set of vertices adjacent to $S_2$. 
The spectral gap of $J(n_2,s_2,m)$ is $\delta=\Omega(\frac{m}{s_2})$.

\paragraph{The Update} To perform an update step on the vertex $S_2$, we use the Element Distinctness algorithm of \cite{amb04} as a subroutine, with some difference in how we define the marked set.
Specifically, we use the subroutine to look for $m$ $2$-collisions, with $m \gg 1$. Furthermore, we  only want to find $2$-collisions that are not already in $S_2$, so $P^{S_2}$ is a walk on $ J(2n/3-2s_2,s_1)$, with vertices corresponding to sets of $s_1$ indices from $(A_1\cup A_2)\setminus\inds(S_2)$,
and we consider a vertex marked if it contains at least $m$ pairs of indices that are $2$-collisions (i.e.,
$\textstyle M^{S_2}=\{S_1\in\binom{(A_1 \cup A_2)\setminus \inds(S_2)}{s_1}: \abs{\cols(S_1)}\geq m\}$). 

\paragraph {The Data} We store the value $\inp_i$ with each $(i,j)\in S_2$ and $i\in S_1$, i.e., $\ket{C(S_2)}=\ket{Q(S_2)}$ and $\ket{d^{S_2}(S_1,S_1')}=\ket{Q(S_1),Q(S_1')}$.  
Although technically this is part of the data, it is classical and coin-independent, so it is straightforward. Furthermore, since $S_1$ is encoded in $Q(S_1)$ and $S_2$ in $Q(S_2)$, we simply write $\ket{Q(S_1)}$ instead of $\ket{S_1,Q(S_1)}$ and $\ket{Q(S_2)}$ instead of $\ket{S_2,Q(S_2)}$. 

The rest of the data is what is actually interesting.
We use the state $\ket{\pi^{S_2}(M^{S_2})}^0_{d^{S_2}}$ in the following instead of $\ket{\pi^{S_2}(M^{S_2})}_{d^{S_2}}$ since it is easy to map between these two states.
For every $S_2\in X$, let
$$\ket{{D}(S_2,0)} \defeq\ket{Q(S_2),0}\ket{\pi^{S_2}(M^{S_2})}^0_{d^{S_2}}=\ket{Q(S_2)}\frac{1}{\sqrt{\abs{M^{S_2}}}}\sum_{\substack{S_1\in M^{S_2}}}\ket{Q(S_1)},$$
and for every edge $(S_2,S_2')$, let 
$\ket{{D}(S_2,S_2')} \defeq \ket{Q(S_2),Q(S_2')}\ket{\Garbage(S_2,S_2')}$ where
\begin{equation}\label{eq:garbage}\ket{\Garbage(S_2,S_2')}\defeq\sum_{\substack{\tilde S_1\in\binom{(A_1\cup A_2)\setminus \inds(S_2\cup S_2')}{s_1-2m}}}\sqrt{\frac{\binom{n_2-s_2}{m}}{\binom{\abs{\cols(\tilde S_1)}+m}{m}\abs{M^{S_2}}}}\ket{Q(\tilde S_1)}.
\end{equation}
We define $\ket{\Garbage}$ in this way precisely because it is what naturally occurs when we attempt to perform the diffusion.

\subsection{Implementation and Cost Analysis}\label{sec:query-analysis}

We now explain how to implement the walk described at the beginning of this section and analyze the costs of the associated operations.

We have assumed that we have some partition $A_1, A_2, A_3$ of $[n]$, although we actually want to run our algorithm on a random partition. The starting state is a uniform superposition over $s_2$ collision pairs across the bipartition $A_1 \times A_2$. Unfortunately, given $A_1,A_2$, we are unable to construct a valid starting state.  However, we can generate a state-partition pair $(\ket{\pi(A_1,A_2)},A_1,A_2)$ such that the distribution of $A_1,A_2,A_3\defeq[n]\setminus (A_1\cup A_2)$ is sufficiently random, and $\ket{\pi(A_1,A_2)}$ is a starting state for the partition $A_1,A_2, A_3$.

\begin{theorem}[Outer walk setup cost $\mathsf{S}_C$]\label{th:query-setup}
The starting state of the outer walk, $\binom{n_2}{s_2}^{-1/2}\sum_{S_2\in \binom{\colset(A_1,A_2)}{s_2}}\ket{Q(S_2)}$, can be constructed for random variables $A_1,A_2,A_3$ with $|A_1|=|A_2|=|A_3|=n/3$, such that if $\inp$ has a unique $3$-collision $\{i,j,k\}$, then $\Pr((i,j,k)\in A_1 \times A_2 \times A_3) = \Omega (1)$, in $\tilde{O}(s_1+s_2\sqrt{{n}/{s_1}})$ queries. 
\end{theorem}
\begin{proof}
To begin, we choose a random tripartition $\tilde{A}_1, \tilde A_2, \tilde A_3$ of $[n]$ such that $\abs{\tilde A_1}=\frac{n}{3}+s_1-s_2$, $\abs{\tilde A_2}=\frac{n}{3}$, and $\abs{\tilde A_3} = \frac{n}{3}-s_1+s_2$. Our final sets $A_1, A_2, A_3$ are closely related to these sets, but satsify $\abs{A_1}=\abs{A_2}=\abs{A_3}=\frac{n}{3}$. Let $\tilde n_2$ be the number of $2$-collisions across $\tilde A_1 \times \tilde A_2$. We first create a uniform superposition over all subsets of $\tilde{A}_1$ of size $s_1$ along with their query values, $\binom{n/3+s_1-s_2}{s_1}^{-1/2}\sum_{I\in \binom{\tilde A_1}{s_1}}\ket{Q(I)}$, using $O(s_1)$ queries. 

For a set $I\in\binom{\tilde{A}_1}{s_1}$, let $H(I)\subset \tilde A_2$ denote the set $\{j\in \tilde A_2:\exists i\in I, \inp_i=\inp_j\}$ of indices in $\tilde A_2$ colliding with $I$. Next we repeatedly Grover search for indices in $H(I)$. For a uniform $I$, the size of $H(I)$ is roughly $\frac{\tilde n_2 s_1}{n}=\Omega(s_1)$ in expectation; more specifically, for most choices $\tilde A_1$ and $\tilde A_2$, we have $\Pr_I(\abs{H(I)}\in\Omega(s_1))\geq 1-o(1)$. We can therefore consider only the part of the state $\binom{n/3+s_1-s_2}{s_1}^{-1/2}\sum_{I\in\binom{\tilde A_1}{s_1}:\abs{H(I)}\geq \epsilon s_1}\ket{Q(I)}$, for a suitable constant $\epsilon$. 
Thus, we can use Grover search to find and query $s_2$ elements of $H(I)$ in $\tilde O(s_2\sqrt{{n}/{s_1}})$ queries, obtaining a state close to 
$$\binom{n/3+s_1-s_2}{s_1}^{-1/2}\sum_{I\in \binom{\tilde A_1}{s_1}:\abs{H(I)}\geq \epsilon s_1}\ket{Q(I)}\binom{\abs{H(I)}}{s_2}^{-1/2}\sum_{J\in\binom{H(I)}{s_2}}\ket{Q(J)}.$$

For a given $J$, we can partition the set $I$ into two disjoint sets: $I_1$, which contains all elements in $I$ that do not collide with any element in $J$; and $I_2$, which contains elements that do collide with an element in $J$. We can then combine $I_2$ with $J$ to get a set of $s_2$ collision pairs. The full reversible mapping, which costs 0 queries,  is $\ket{Q(I),Q(J)}\mapsto \ket{Q(I_1)}\ket{\{(i,j,\inp_i):i\in I_2, j\in J\}}$. Applying this transformation gives (a state close to)
$$\binom{n/3+s_1-s_2}{s_1}^{-1/2}\sum_{I_1\in\binom{\tilde A_1}{s_1-s_2}:\abs{H(I_1)}\geq \epsilon s_1 - s_2}\ket{Q(I_1)}\binom{\abs{H(I_1)}+s_2}{s_2}^{-1/2}\sum_{S_2\in\cols(\tilde A_1\setminus I_1,\tilde A_2)}\ket{Q(S_2)}.$$
Note that this state is not uniform in $I_1$, but is uniform in $S_2$ when we restrict to a particular $I_1$. Thus we measure the first register to get some $I_1$ with non-uniform probability that depends only on $\abs{H(I_1)}$. The remaining state is the uniform superposition 
$$\binom{\abs{\colset(\tilde A_1\setminus I_1,\tilde A_2)}}{s_2}^{-1/2}\sum_{S_2\in\binom{\colset(\tilde A_1\setminus I_1,\tilde A_2)}{s_2}}\ket{Q(S_2)}.$$

Now let $A_1=\tilde{A}_1\setminus I_1$, $A_2 = \tilde A_2$ and $A_3= \tilde{A_3} \cup I_1$. Then we have $\colset = \colset(\tilde{A}_1\setminus I_1,\tilde A_2) = \colset(A_1,A_2)$, so we have constructed the correct state for the tripartition $A_1,A_2,A_3$. Clearly, if $\{i,j,k\}$ is the unique $3$-collision, then $i\in \tilde A_1$, $j\in \tilde A_2$ and $k\in \tilde A_3$ with constant probability. It remains to consider whether $i \in I_1$. Although the distribution of $I_1$ is non-uniform, the distribution restricted to those $I_1$ with $H(I_1)=h$ is uniform for any fixed $h$, and it is easy to see that $\Pr(i\in I_1|H(I_1)=h)$ is $o(1)$ for any $h$. 

For more details, refer to the proof of Theorem \ref{th:setup-time}, which also proves an analogous statement for time complexity.
\end{proof}

Hereafter, we assume the above choice of partition $A_1,A_2,A_3$, and that if there is a unique $3$-collision $\{i,j,k\}$, then $i\in A_1$, $j\in A_2$ and $k\in A_3$.

\begin{theorem}[Costs of the update walk $\mathsf{S}',\frac{1}{\sqrt{\eps'}}(\frac{1}{\sqrt{\delta'}}\mathsf{U}'+\mathsf{C}')$]\label{th:query-nested-cost}
The update walk has query complexities
$\mathsf{S}'={O}(s_1)$ and $\frac{1}{\sqrt{\eps'}}(\frac{1}{\sqrt{\delta'}}\mathsf{U}'+\mathsf{C}')=\tilde{O}(\sqrt{{nm}/{s_1}}).$
\end{theorem}
\begin{proof}
Fix an arbitrary vertex $S_2\in \binom{\colset(A_1,A_2)}{s_2}$. We now analyze the update walk $P^{S_2}$. 
The walk is still on $J(2n/3-2s_2,s_1)$, so $\delta'=\Omega(\frac{1}{s_1})$, but in contrast to \cite{amb04}, a vertex is considered marked if it has at least $m$ collision pairs, and we have a lower bound of $n_2=\Omega(n)$ on the number of disjoint collision pairs. Since we defined $m={s_1^2n_2}/{n^2}$ as roughly the  expected number of collision pairs in a set of size $s_1$, we have $\eps'=\Omega(1)$. We still need to do the walk, both to amplify the success probability to inverse polynomial (which we could also have done by increasing $s_1$ by log factors) and more importantly, to implement the phase flip $\ket{x,0}\ket{D(x,0)}\mapsto -\ket{x,0}\ket{D(x,0)}$.  

\begin{description}
\item[\emph{Setup}:] We need $\mathsf{S}'={O}(s_1)$ queries to set up $\binom{2n/3-2s_2}{s_1}^{-1/2} \sum_{S_1\in\binom{(A_1\cup A_2)\setminus \inds(S_2)}{s_1}}\ket{Q(S_1)}$.

\item[\emph{Update}:] The update on $J(2n/3-2s_2,s_1)$ costs $O(1)$ queries. 

\item[\emph{Checking}:] The query complexity of checking is $0$, since we merely observe whether there are $m$ colliding pairs in $S_1$. 
\end{description}

We can thus compute $\frac{1}{\sqrt{\eps'}}(\frac{1}{\sqrt{\delta'}}\mathsf{U}'+\mathsf{C}')=\tilde{O}(\sqrt{\frac{mn^2}{s_1^2n_2}}\sqrt{s_1})=\tilde{O}(\sqrt{\frac{{nm}}{s_1}})$. 
\end{proof}

The following lemma tells us that we do not need to reverse the garbage part of the data.

\begin{lemma}\label{lem:garbage}
For all edges $(S_2,S_2')$, $\ket{\Garbage(S_2,S_2')}=\ket{\Garbage(S_2',S_2)}$.
\end{lemma}

\begin{proof}
Recall the definition of $\ket{\Garbage(S_2,S_2')}$ from (\ref{eq:garbage}):
$$\ket{\Garbage(S_2,S_2')}=\sum_{\substack{\tilde S_1\in\binom{(A_1\cup A_2)\setminus \inds(S_2\cup S_2')}{s_1-2m}}}\sqrt{\frac{\binom{n_2-s_2}{m}}{\binom{\abs{\cols(\tilde S_1)}+m}{m}\abs{M^{S_2}}}}\ket{Q(\tilde S_1)}.$$

To see that this is symmetric in $S_2$ and $S_2'$, we need only show that $\abs{M^{S_2}}=\abs{M^{S_2'}}$. 
We have 
$$\abs{M^{S_2}}=\left|\left\{S_1\in\binom{(A_1\cup A_2)\setminus\inds(S_2)}{s_1}:\abs{\cols(S_1)}\geq m\right\}\right|=\binom{\abs{\cols\setminus S_2}}{m}\binom{\abs{(A_1\cup A_2)\setminus \inds(S_2)}-2m}{s_1-2m}.$$
This holds because all collisions in $A_1\cup A_2$ are disjoint, and so choosing $m$ pairs from $\cols\setminus S_2$ gives $2m$ distinct indices. We can easily see that $\abs{\cols\setminus S_2}=n_2-s_2$, which is independent of $S_2$. Less trivially, since all collisions in $S_2$ are disjoint, we have $\abs{\inds(S_2)}=2s_2$ for all $S_2$, and so $\abs{(A_1\cup A_2)\setminus \inds(S_2)}=\abs{(A_1\cup A_2)}-2s_2$, again, independent of $S_2$. Thus we have $\abs{M^{S_2}}=\abs{M^{S_2'}}$, completing the proof.
\end{proof}

From this lemma it readily follows that the \textsc{Garbage Swap} requires no queries.

\begin{theorem}[\textsc{Garbage Swap} cost]\label{th:query-db-swap}
No queries are needed to perform the \textsc{Garbage Swap}, which for any $(S_2,S_2')$ performs the map
$$\ket{Q(S_2),Q(S_2')}\ket{\Garbage(S_2,S_2')}\mapsto \ket{Q(S_2'),Q(S_2)}\ket{\Garbage(S_2',S_2)}.$$
\end{theorem}

The \textsc{Local Diffusion with Garbage} also requires no queries, but is nontrivial to implement.

\begin{theorem}[\textsc{Local Diffusion with Garbage} cost]\label{th:query-ld-cost}
No queries are needed to perform the \textsc{Local Diffusion with Garbage}, which, for any $S_2$, performs the map
$$\ket{Q(S_2)}\ket{\pi^{S_2}(M^{S_2})}^0 \mapsto \frac{1}{\sqrt{\abs{\Gamma(S_2)}}}\sum_{S_2'\in\Gamma(S_2)}\ket{Q(S_2),Q(S_2')}\ket{\Garbage(S_2,S_2')},$$
where $\ket{\Garbage(S_2,S_2')}$ is defined in (\ref{eq:garbage}).

\end{theorem}
\begin{proof}
We employ the following procedure to perform the \textsc{Local Diffusion with Garbage}.
\begin{enumerate}
\item Perform $\ket{Q(S_2),Q(S_1)}\mapsto \ket{Q(S_2),Q(S_1)}\binom{s_2}{m}^{-1/2}\sum_{I\in\binom{S_2}{m}}\ket{Q(I)}$. 
\item Perform $\ket{Q(S_2),Q(S_1)}\mapsto \ket{Q(S_2),Q(S_1)}\binom{\abs{\cols(S_1)}}{m}^{-1/2}\sum_{J\in \binom{\cols(S_1)}{m}}\ket{Q(J)}$.
\item Perform $\ket{Q(S_1)}\ket{Q(J)}\mapsto \ket{Q(\tilde S_1)}\ket{Q(J)}$, where $\tilde S_1=S_1\setminus\inds(J)$. 
\end{enumerate}
Here $I$ represents collision pairs to be removed from $S_2$ and $J$ represents collision pairs to be added. 
It is clear that each of these operations has query complexity $0$. 

Now we show the correctness of this procedure. Recall that $S_1\in M^{S_2}$ if and only if $S_1$ contains $m$ collisions, i.e., $\abs{\cols(S_1)}\geq m$, so $\ket{Q(S_2)}\ket{\pi^{S_2}(M^{S_2})}^0$ is
$$\ket{Q(S_2)}\frac{1}{\sqrt{\abs{M^{S_2}}}}\sum_{\substack{S_1\in M^{S_2}}}\ket{Q(S_1)} = \ket{Q(S_2)}\frac{1}{\sqrt{\abs{M^{S_2}}}}\sum_{S_1\in\binom{(A_1\cup A_2)\setminus \inds(S_2)}{s_1}:\abs{\cols(S_1)}\geq m}\ket{Q(S_1)}.$$
After performing the above procedure, we get the state
$$\ket{Q(S_2)}
\frac{1}{\sqrt{\abs{M^{S_2}}}}\sum_{\substack{S_1\in\binom{(A_1\cup A_2)\setminus \inds(S_2)}{s_1}:\\\abs{\cols(S_1)}\geq m}}\!\!\!\!\!\ket{Q(\tilde S_1)}
\frac{1}{\sqrt{\binom{s_2}{m}}}\sum_{\substack{I\in\binom{S_2}{m}}}
\ket{Q(I)}
\frac{1}{\sqrt{\binom{\abs{\cols(S_1)}}{m}}} \sum_{\substack{J\in \binom{\cols(S_1)}{m}}}\!\!\!\!\!\ket{Q(J)}.$$
Note that $\cols(S_1)=\cols(\tilde S_1)\cup J$. To see this, we must appeal to the fact that all collisions in $A_1\times A_2$ are disjoint, by assumption, so for each collision pair $(i,j)\in J$, removing $i$ and $j$ from $S_1$ removes the collision pair $(i,j)$ and no other collision pair from $\cols(S_1)$. Next, we can see that $\abs{\cols(\tilde S_1)\cup J}=\abs{\cols(\tilde S_1)}+m$, since $J\cap \tilde \cols(S_1) = \emptyset$ and $\abs{J}=m$. 
Thus, $\abs{\cols(S_1)}=\abs{\cols(\tilde S_1)}+m$, and we can rewrite the state as
$$\ket{Q(S_2)}\binom{s_2}{m}^{-1/2}\binom{n_2-s_2}{m}^{-1/2}\!\!\!\!\!\sum_{\substack{I\in\binom{S_2}{m}\\ J\in \binom{\cols\setminus S_2}{m}}}\!\!\!\!\!\!\ket{Q(I)}\ket{Q(J)}\!\!\sum_{\substack{\tilde S_1\in\binom{(A_1\cup A_2)\setminus \inds(S_2\cup J)}{s_1-2m}}}
\!\!\!\!\alpha_{\tilde S_1}(S_2,(S_2\cup J)\setminus I)\ket{Q(\tilde S_1)},$$
where 
\[
\alpha_{\tilde S_1}(S_2,(S_2\cup J)\setminus I)=\sqrt{\frac{\binom{n_2-s_2}{m}}{\binom{\abs{\cols(\tilde S_1)}+m}{m}\abs{M^{S_2}}}}.
\]
We now simply note that the neighbours of any $S_2\in X$ are exactly $(S_2\cup J)\setminus I$ for $I\in\binom{S_2}{m}$ and $J\in\binom{\colset\setminus S_2}{m}$. Furthermore, for such a neighbour $S_2'=(S_2\cup J)\setminus I$, $Q(S_2),Q(I),Q(J)$ encodes $Q(S_2),Q(S_2')$. Finally, for such an $S_2'$, we have $S_2\cup J=S_2\cup S_2'$. Thus, we are left with the desired state. 
\end{proof}

\begin{corollary}[\textsc{Local Diffusion} and \textsc{Database Swap} cost $\mathsf{T}$]\label{cor:query-ld}
The family $(P^{S_2},M^{S_2},d^{S_2})_{S_2\in X}$ implements the \textsc{Local Diffusion} and \textsc{Database Swap} of $(P,Q)$ with no queries.
\end{corollary}
\begin{proof}
This is immediate from Theorems \ref{th:query-db-swap} and \ref{th:query-ld-cost}.
\end{proof}

The checking cost is immediate, since we can use Grover search to look for an element of $A_3$ that collides with any of the stored 2-collisions.

\begin{theorem}[Checking cost $\mathsf{C}$]\label{th:query-checking-cost}
We can implement the checking reflection with $\mathsf{C}=\tilde{O}(\sqrt{n})$ queries.
\end{theorem}

We now have all necessary ingredients to prove the main theorem.

\begin{proof}[Proof of Theorem \ref{th:query-main}]

We apply Theorem \ref{th:nested-update} to compute the cost of our nested-update quantum walk algorithm, giving (up to log factors)
\begin{align*}
&\textstyle\mathsf{S}_C+\mathsf{S}'+\frac{1}{\sqrt{\eps}}\left(\frac{1}{\sqrt{\delta}}\left(\frac{1}{\sqrt{\eps'}}\left(\frac{1}{\sqrt{\delta'}}\mathsf{U}'+\mathsf{C}'\right)+\mathsf{T}\right)+\mathsf{C}\right) \\
&\quad\textstyle=s_1+s_2\sqrt{\frac{n}{s_1}}+s_1+\sqrt{\frac{n_2}{s_2}}\left(\sqrt{\frac{s_2}{m}}\left(\frac{\sqrt{nm}}{\sqrt{s_1}}+0\right)+\sqrt{n}\right) \\
&\quad\textstyle= s_1+s_2\sqrt{\frac{n}{s_1}}+\sqrt{\frac{n_2n}{s_1}}+\sqrt{\frac{n_2n}{s_2}}
=\tilde{O}\left(s_1+s_2\sqrt{\frac{n}{s_1}}+\frac{n}{\sqrt{s_1}}+\frac{n}{\sqrt{s_2}}\right)
\end{align*}
using the cost calculations from Theorems \ref{th:query-setup}, \ref{th:query-nested-cost}, and \ref{th:query-checking-cost} and Corollary \ref{cor:query-ld}. Setting $s_1=n^{5/7}$ and $s_2=n^{4/7}$ gives query complexity $\tilde{O}(n^{5/7})$.
\end{proof}

\section{Time Complexity of 3-Distinctness}\label{sec:time}

In this section we prove the following theorem.

\begin{theorem}\label{th:main}
The time complexity of $3$-Distinctness is $\tilde{O}(n^{5/7})$. 
\end{theorem}

\noindent This follows fairly straightforwardly from the quantum walk described in Section \ref{sec:query}. The only remaining task is to describe how we can encode the sets of queried indices and pairs of indices so that all necessary operations, such as inserting an element in a set or removing an element from a set, can be done in poly-logarithmic time. We use the same data structure that was used to obtain a tight upper bound on the time complexity of Element Distinctness \cite{amb04}. After describing the necessary properties of this data structure and how we apply it to our walk, we explain how each of the operations described in Section \ref{sec:query} can be done time-efficiently using this encoding.

The following lemma describes properties of the data structure that we use to encode edges of our walk and their data.  We refer to this data structure as a \emph{skip-list}.

\begin{lemma}[\cite{amb04}]\label{lem:data}
There exists a data structure for storing a set of items of the form $(z,\inp)$ (the $\inp$ values need not be unique) that allows the following operations to be performed in worst case time complexity $O(\log^4(n+q))$:
insert an item; 
delete an item;
look up an item by its $\inp$ value; or
create a superposition of the elements stored.
The data structure storing a set $S$ is a unique encoding of $S$. 
\end{lemma}

\paragraph{Encoding an Edge and its Data}  We now describe how to encode an edge and its data.  
These states have the form either $\ket{S_2,S_2',D(S_2,S_2')}$, where $\ket{D(S_2,S_2')}$ is a superposition over basis states $|Q(S_2),Q(S_2'),Q(\tilde S_1)\>$ for $(S_2,S_2')\in E$ and $\tilde S_1\subset (A_1\cup A_2)\setminus \inds( S_2\cup S_2')$ (and recall that $Q(S_2),Q(S_2')$ automatically encodes $S_2,S_2'$); or $\ket{S_2,0,D(S_2,0)}$, where $\ket{D(S_2,0)}$ is a superposition over basis states $\ket{Q(S_2),Q(S_1),Q(S_1')}$ for $S_2\in X$, and $(S_1,S_1')\in E^{S_2}\cup (V^{S_2}\times\{0\})$.
Strictly speaking, we previously defined $\ket{D(S_2,0)}$ using  $\ket{\pi^{S_2}(M^{S_2})}^0_{d^{S_2}}$ instead of $\ket{\pi^{S_2}(M^{S_2})}_{d^{S_2}}$, that is, as a superposition of $\ket{Q(S_2),Q(S_1),0}$ for $S_2\in X$, and $S_1\in V^{S_2}$. 
However, we must also consider how to encode states of the nested update walk, which do not generally have $0$ in the coin register.

We begin by encoding the triple of sets $(Q(S_2),Q(S_1),Q(S_1'))$. We store each of $Q(S_2)$ and $Q(S_1)$ in a skip-table. To store $Q(S_1')$ for $S_1'\neq 0$, we simply store both $Q(S_1)\setminus Q(S_1')$ and $Q(S_1')\setminus Q(S_1)$, each of which is a single queried index $(i,\inp_i)$. This already encodes the three sets, but we add additional structure to speed up certain tasks. We store $Q(\cols(S_1))$, the set of $2$-collisions in $S_1$, in another skip-table.  We also keep a counter of the size of this set so that we can easily check whether $S_1$ is marked.

Now we describe how we encode the triple of sets $(Q(S_2),Q(S_2'),Q(\tilde S_1))$. We store each of $Q(S_2)$ and $Q(\tilde S_1)$ in a skip-table. To store $Q(S_2')$, we store each of $Q(S_2)\setminus Q(S_2')$ and $Q(S_2')\setminus Q(S_2)$ in a skip-table. We also store with $Q(\tilde S_1)$ 
an additional skip-table containing $Q(\colset(\tilde S_1))$, with a counter encoding its size. We store this simply because this is what is left over from the encoding of $Q(S_1)$ after we perform \textsc{Local Diffusion}.

It is now clear from Lemma \ref{lem:data} that we can perform the following operations in poly-logarithmic time: insert an element to $S_1$, insert an element to $S_2$, delete an element from $S_1$, delete an element from $S_2$, look up an element in $S_1$, and look up an element in $S_2$. In addition, we can perform each of these operations in superposition.

\paragraph{Cost Analysis} 
We now analyze the cost of all operations implemented in Section \ref{sec:query}.

Since we are now concerned with time complexity, we need some efficient way to store and compute the partition $A_1,A_2,A_3$. We use the following notion to efficiently represent a random subset of $[n]$.
\begin{definition}
A family $F$ of functions $f:[n]\to [\ell]$ is said to be \emph{$k$-wise independent} if
for any distinct $i_1,\ldots,i_k\in [n]$, the distribution $(f(i_1),\ldots,f(i_k))$ is identical to the uniform distribution over $[\ell]^k$.
\end{definition}
When $n=\ell$ is a prime power, a simple example of a $k$-wise independent functions is the family of all polynomials of degree $k-1$ over the finite field $\mathrm{GF}(n)$~\cite{wc81}.
Each such polynomial can be represented using $O(k\log n)$ bits and evaluated using
$O(k)$ additions and multiplications over $\mathrm{GF}(n)$. More efficient constructions exist, but the polynomial construction suffices here.
It can be extended to any integer $n$ by allowing a small statistical distance from
the distribution $(f(i_1),\ldots,f(i_k))$ to the uniform distribution over $[n]^k$.

For simplicity, we now assume that we have at our disposal a (perfect) $3$-wise independent family $F$ of functions $f:[n]\to[n]$.

\begin{theorem}[Outer walk setup cost $\mathsf{S}_C$]\label{th:setup-time}
We can construct,  in time $\tilde{O}(s_1+s_2\sqrt{{n}/{s_1}})$, a state $\binom{n_2}{s_2}^{-1/2}\sum_{S_2\in \binom{\colset(A_1,A_2)}{s_2}}\ket{Q(S_2)}$ for $A_1,A_2$ random variables such that
\begin{enumerate} 
\item $\abs{A_1}=\abs{A_2}=\frac{n}{3}$;
\item $A_1,A_2,A_3\defeq[n]\setminus (A_1\cup A_2)$ is a tripartition of $[n]$;
\item  if $\inp$ has a unique $3$-collision $\{i,j,k\}$, $\Pr(i\in A_1,j\in A_2,k\in A_3)=\Omega(1)$;
\item the space complexity of storing the partition $(A_1,A_2,A_3)$ is $\tilde{O}(s_1)$; and
\item the time complexity of determining to which of $A_1$, $A_2$ or $A_3$ an index $i\in [n]$ belongs is $\tilde{O}(1)$.
\end{enumerate}
\end{theorem}
\begin{proof} This proof is similar to that of Theorem \ref{th:query-setup}, but we include significantly more detail. 
Let $f\in F$ be a $3$-wise independent function, and define $\tilde A_1$, $\tilde A_2$ and $\tilde A_3$ by $f(i)\leq \frac{n}{3}+s_1-s_2\Leftrightarrow i\in \tilde A_1$, $\frac{n}{3}+s_1-s_2 < f(i)\leq \frac{2n}{3}+s_1-s_2 \Leftrightarrow i\in\tilde A_2$, and $f(i)>\frac{2n}{3}+s_1-s_2\Leftrightarrow i\in \tilde A_3$. Then $\tilde A_1,\tilde A_2,\tilde A_3$ is a partition of $[n]$ with $\abs{\tilde A_1}=\frac{n}{3}+s_1-s_2$, $\abs{\tilde A_2}=\frac{n}{3}$, and $\abs{\tilde A_3}=\frac{n}{3}-s_1+s_2$.

If $\inp$ has a unique $3$-collision $\{i,j,k\}$, then the $3$-wise independence of $f$ implies that
\begin{align*}
&  \Pr_f\(i\in \tilde A_1,j\in\tilde A_2,k\in\tilde A_3\)\\
 &\quad= \Pr_f\(f(i)\leq \frac{n}{3}+s_1-s_2,\frac{n}{3}+s_1-s_2 < f(j)\leq \frac{2n}{3}+s_1-s_2,f(k)>\frac{2n}{3}+s_1-s_2\)\\
&\quad\geq \left(\frac{1}{3}-o(1)\right)^3=\frac{1}{27}-o(1).
\end{align*}
We assume that this holds when constructing the starting state. Otherwise, our construction fails and we try again. Let $\tilde{n}_2 \defeq \abs{\colset(\tilde A_1,\tilde A_2)}$. We can assume that $\tilde{n}_2\in\Omega(n)$ for the same reason we can always assume that $\chi$ has $\Omega(n)$ $2$-collisions. 

To begin, we create a uniform superposition 
$
  \binom{\abs{\tilde A_1}}{s_1}^{-1/2}\sum_{I\in\binom{\tilde A_1}{s_1}}\ket{Q(I)}
$
of sets of $s_1$ indices drawn from $\tilde{A}_1$, stored in a skip-table, and query these indices.  This uses $s_1$ queries and $s_1$ insertions, for a total time complexity of $\tilde{O}(s_1)$.

For $I\subset \tilde A_1$, let $H(I)\defeq\{j\in \tilde A_2:\exists i\in I,\inp_j=\inp_i\}$. Next, we search $\tilde A_2$ for indices in $H(I)$, assuming that $H(I)$ has size at least $\Omega(\frac{\tilde n_2s_1}{n})$. The following lemma justifies this assumption.

\begin{lemma}
Let $I$ be a uniformly random subset of $\tilde A_1$ of size $s_1$. Then $\Pr_I(\abs{H(I)} \leq \frac{\tilde{n}_2}{n} s_1)< o(1)$. 
\end{lemma}
\begin{proof}
The random variable $\abs{H(I)}$ has a hypergeometric distribution with mean $\mu=\frac{s_1\tilde{n}_2}{2n/3+s_1-s_2}= \Theta(s_1)$. Using tail inequalities \cite[eq.~14]{ska11} we have, for any constant $c\geq 1$,
\[
\Pr(\abs{H(I)}\leq \tfrac{1}{c}\mu)\leq \exp\left(-2\left(\frac{\mu(1-\tfrac{1}{c})}{s_1}\right)^2s_1\right) 
\leq e^{-\Theta(s_1)}=o(1). \qedhere
\]
\end{proof}

\noindent Thus we can then restrict our attention to the part of the state
with $\abs{H(I)}\geq \epsilon s_1$
for some constant $\epsilon \leq \frac{\tilde n_2}{n}$, as this part has $1-o(1)$ of the weight. 
We can then perform the mapping 
$$\binom{|\tilde A_1|}{s_1}^{-1/2}\!\!\!\sum_{I\in\binom{\tilde A_1}{s_1}:\abs{H(I)}\geq \epsilon s_1}\ket{Q(I)} \mapsto \binom{|\tilde{A}_1|}{s_1}^{-1/2}\!\!\!\sum_{I\in\binom{\tilde A_1}{s_1}:\abs{H(I)}\geq \epsilon s_1}\ket{Q(I)}\binom{\abs{H(I)}}{s_2}^{-1/2}\!\!\!\sum_{J\in\binom{H(I)}{s_2}}\ket{Q(J)}$$
using $s_2$ applications of Grover search for a new element of $H(I)$. Each search requires $\tilde{O}(\sqrt{{n}/{\abs{H(I)}}})=\tilde O(\sqrt{n/s_1})$ iterations.  As we find elements, we insert them into a skip-table, also separately recording the order in which we find the indices. To check if some $i$ is in $H(I)$, but not already found, we 
\begin{itemize}
\item look up $i$ in the skip-table of indices already found;
\item query $\inp_i$;
\item compute $f(i)$; and
\item look up $\inp_i$ in $Q(I)$.
\end{itemize}
Each of these operations has time complexity $\tilde{O}(1)$, for a total cost of $\tilde{O}(1)$ per iteration.
The total cost of the $s_2$ rounds of search is $\tilde{O}(s_2\sqrt{n/s_1})$. Finally, we must uncompute the order in which we found the indices of $J$. For each $J$, we have the state $\ket{Q(J)}2^{-s_2/2}\sum_{\sigma\in \mathcal{S}_{s_2}}\ket{\sigma(J)}$, where $\mathcal{S}_n$ is the symmetric group on $n$ symbols. We can uncompute the order register in cost $\tilde{O}(s_2)$, completing the desired mapping.

Let $I_1$ be the elements of $I$ for which we did not find a collision, and $I_2$ those elements of $I$ for which we did find a collision. We can reversibly convert $\ket{Q(I),Q(J)}$ to $\ket{Q(I_1),\{(i,j,\inp_i):i\in I_2,j\in J,\inp_i=\inp_j\}}$, where both sets are stored in a skip-table. We call the second set $Q(S_2)$. To accomplish this mapping, we do the following $s_2$ times, once for each $j\in J$: 
\begin{itemize}
\item look up $\inp_j$ in $Q(I)$ to find $(i,\inp_i=\inp_j)$;
\item insert $(i,j,\inp_i)$ into $Q(S_2)$; and
\item delete $(i,\inp_i)$ from $I$ and $(j,\inp_i)$ from $J$.
\end{itemize}
What remains in $Q(I)$ after performing these steps is exactly $Q(I_1)$. Each repetition costs $\tilde{O}(1)$, for a total cost of $\tilde{O}(s_2)$. Note that we can delete $(i,\inp_i)$ every time because all $2$-collisions in $\tilde{A}_1\times \tilde{A}_2$ are disjoint (by assumption), so $\abs{I_1}=s_1-s_2$. We also have $\abs{H(I_1)}=\abs{H(I)}-s_2$, again because all $2$-collisions in $\tilde A_1\times \tilde{A}_2$ are disjoint. 
Thus, after performing the full mapping, the part of the state under consideration is 
$$\binom{|\tilde A_1|}{s_1}^{-1/2}\sum_{I_1\in\binom{\tilde A_1}{s_1-s_2}}\binom{\abs{H(I_1)}+s_2}{s_2}^{-1/2}\ket{Q(I_1)}\sum_{S_2\in\binom{\colset(\tilde A_1\setminus I_1,\tilde A_2)}{s_2}}\ket{Q(S_2)}.$$
Measuring the first register, containing some $\ket{Q(I_1)}$, gives the state
$$\binom{|\colset(\tilde A_1\setminus I_1, \tilde A_2)|}{s_2}^{-1/2}\sum_{S_2\in\binom{\colset(\tilde A_1\setminus I_1,\tilde A_2)}{s_2}}\ket{Q(S_2)}$$
for some $I_1$ with probability at least $1-o(1)$. 
Adding up the total cost, we find $\mathsf{S}_C=\tilde{O}(s_1+s_2\sqrt{{n}/{s_1}}+s_2)=\tilde{O}(s_1+s_2\sqrt{{n}/{s_1}})$, since $n>s_1$. 

This state is the correct starting state for the partition $A_1=\tilde A_1\setminus I_1, A_2=\tilde A_2,A_3=\tilde A_3\cup I_1$, which is clearly a tripartition with $\abs{A_1}=\abs{A_2}=\abs{A_3}=n/3$. Furthermore, for a $3$-collision $\{i,j,k\}$, assuming $i\in\tilde A_1,j\in\tilde A_2,k\in\tilde A_3$ (which happens with constant probability), the only way we can fail to have $i\in A_1,j\in A_2, k\in A_3$ is if $i\in I_1$. Although $I_1$ is not uniformly distributed, $\Pr(I_1|H(I_1)=h)$ is uniform for any $h$.  Furthermore, $\Pr(i\in I_1|H(I_1)=h)=o(1)$ for any fixed $h$, since $\abs{I_1} \ll n$. Thus, we have $\Pr_{A_1,A_2,A_3}(i\in A_1,j\in A_2,k\in A_3)=\Omega(1)$.

Finally, to store the tripartition $A_1,A_2,A_3$, we need to keep $f$, 
 as well as $I_1$, which we store in a skip-table. This takes space $\tilde{O}(1)+\tilde{O}(\abs{I_1})=\tilde{O}(s_1)$. To compute which of $A_1$, $A_2$, or $A_3$ contains an index $i$, we first compute $f(i)$, and then (possibly) look up $i$ in $I_1$, each of which costs $\tilde{O}(1)$. 
\end{proof}

\begin{theorem}[Costs of the update walk $\mathsf{S}',\frac{1}{\sqrt{\eps'}}(\frac{1}{\sqrt{\delta'}}\mathsf{U}'+\mathsf{C}')$]\label{th:time-nested-cost}
The update walk has 
time complexities
$\mathsf{S}'=\tilde{O}(s_1)$ and $\frac{1}{\sqrt{\eps'}}(\frac{1}{\sqrt{\delta'}}\mathsf{U}'+\mathsf{C}')=\tilde{O}(\sqrt{{{nm}}/{s_1}}).$
\end{theorem}

\begin{proof}
This follows from Theorem \ref{th:query-nested-cost} and our encoding of a triple $(Q(S_2),Q(S_1),Q(S_1'))$. The implementation is nearly identical to the time-efficient Element Distinctness algorithm of \cite{amb04}, except that we store  
an extra skip-table containing the set $Q(\cols(S_1))$. However, insertion and deletion may still be performed in poly-logarithmic time. To insert $i$ into $S_1$, we must look up $\inp_i$ in $Q(S_1)$ to see if we have a new collision in $Q(S_1)$, i.e., if there is some $(j,\inp_i)$ already in $Q(S_1)$ such that $(i,j)\in (A_1\times A_2)\cup (A_2\times A_1)$. If there is such a $j$, then we insert $(i,j,\inp_i)$ into $Q(\cols(S_1))$ if $(i,j)\in A_1\times A_2$ or $(j,i,\inp_i)$ into $Q(\cols(S_1))$ else. Finally, we  insert $(i,\inp_i)$ into $Q(S_1)$. 
This involves a constant number of skip-table insertions and lookups, so its cost is still poly-logarithmic. We can delete some $i\in S_1$ by running this operation in reverse. Thus, the update and setup cost are clearly the same as \cite{amb04}. From the proof of Theorem \ref{th:query-nested-cost}, we have
$$\mathsf{S}'=\tilde{O}(s_1)\quad\quad\mathsf{U}'=\tilde{O}(1)\quad\quad\delta'=\Omega(1/s_1)\quad\quad\eps'=\Omega(1).$$

To check if $S_1$ is marked, we simply read the counter storing the size of $Q(\cols(S_1))$ and check if it is at least $m$, in time $\mathsf{C}'=O(1)$. Thus, we have
\[
\mathsf{S}'=\tilde{O}(s_1)\quad\quad \frac{1}{\sqrt{\eps'}}\left(\frac{1}{\sqrt{\delta'}}\mathsf{U}'+\mathsf{C}'\right)=\tilde{O}\left(\sqrt{\frac{nm}{s_1}}\right). \qedhere
\]
\end{proof}

Since $\ket{\Garbage(S_2,S_2')}=\ket{\Garbage(S_2',S_2)}$ for all edges $(S_2,S_2')$ by Lemma \ref{lem:garbage}, we have the following. 

\begin{theorem}[\textsc{Garbage Swap} cost]\label{th:db-swap-cost}
We can implement the \textsc{Garbage Swap} in time $\tilde{O}(m)$.
\end{theorem}
\begin{proof}
The \textsc{Garbage Swap} is the operation that acts, for any edge $(S_2,S_2')$, as 
$$\ket{Q(S_2),Q(S_2')}\ket{\Garbage(S_2,S_2')}\mapsto \ket{Q(S_2'),Q(S_2)}\ket{\Garbage(S_2',S_2)}.$$
By Lemma \ref{lem:garbage}, we need only consider the cost of $\ket{Q(S_2),Q(S_2')}\mapsto \ket{Q(S_2'),Q(S_2)}$. Recall that $\ket{Q(S_2),Q(S_2')}$ is encoded as $\ket{Q(S_2)}\ket{Q(S_2\setminus S_2')}\ket{Q(S_2'\setminus S_2)}$, with each of the three parts encoded as a skip-table. Since $(S_2,S_2')$ is an edge, we have $\abs{S_2\setminus S_2'}=\abs{S_2'\setminus S_2}=m$. Thus, we can perform the mapping to $\ket{Q(S_2')}\ket{Q(S_2'\setminus S_2)}\ket{Q(S_2\setminus S_2')}$ 
 by performing $m$ insertions and $m$ deletions on $Q(S_2)$ to get $Q(S_2')$. 
\end{proof}

\begin{theorem}[\textsc{Local Diffusion with Garbage} cost]\label{th:ld-cost}
We can implement the \textsc{Local Diffusion with Garbage} with time complexity $\tilde{O}(m)$. 
\end{theorem}
\begin{proof}
We consider each of the three steps from the proof of Theorem \ref{th:query-ld-cost}. 
In step 1, we create a superposition over sets of $m$ values in $Q(S_2)$ using $m$ superposition accesses to the skip-table storing $Q(S_2)$. In step 2 we create a superposition over sets of $m$ values in $Q(\cols(S_1))$ using $m$ superposition accesses to the skip-table storing $Q(\cols(S_1))$. Finally, in step 3 we perform $m$ lookups in the skip-table storing $Q(J)=Q(S_2'\setminus S_2)$ and $2m$ deletions from the skip-table storing $Q(S_1)$. The total cost of this is $\tilde{O}(m)$. 

It is also clear from the proof of Theorem \ref{th:query-ld-cost} that we move from a superposition of correctly encoded triples $(Q(S_2),Q(S_1),0)$ (where the 0 corresponds to the coin of $S_1$) to a superposition of correctly encoded triples $(Q(S_2),Q(S_2'),Q(\tilde S_1))$. 
\end{proof}

\begin{corollary}[\textsc{Local Diffusion} and \textsc{Database Swap} cost $\mathsf{T}$]\label{cor:ld}
The family $(P^{S_2},M^{S_2},d^{S_2})_{S_2\in X}$ implements the \textsc{Local Diffusion} and \textsc{Database Swap} of $(P,Q)$ with time complexity $\mathsf{T}=\tilde{O}(m)$.
\end{corollary}
\begin{proof}
This is immediate
from Theorems \ref{th:db-swap-cost} and \ref{th:ld-cost}.
\end{proof}

\begin{theorem}[Checking cost $\mathsf{C}$]\label{th:checking-cost}
We can implement the checking reflection in  time $\mathsf{C}=\tilde{O}(\sqrt{n})$.
\end{theorem}
\begin{proof}
To check if a vertex $S_2$ is marked, we search for an index $k\in A_3$ such that there exists $(i,j)\in S_2$ such that $\{i,j,k\}$ is a $3$-collision. Each time we check if a particular $k$ has this property, we query $k$ and look up $\inp_k$ in $Q(S_2)$ in time $\tilde{O}(1)$. 
\end{proof}

Theorem \ref{th:main} now follows. All costs are the same as their query complexities from Section \ref{sec:query}, with the exception of $\mathsf{T}=\tilde{O}(m)=\tilde{O}(\frac{s_1^2}{n})$, but this does not change the asymptotic complexity. Plugging the values from Theorems \ref{th:setup-time}, \ref{th:time-nested-cost}, and \ref{th:checking-cost} and Corollary \ref{cor:ld} into the nested update cost expression from Theorem \ref{th:nested-update}, we have
\begin{align*}
&\textstyle\mathsf{S}_C+\mathsf{S}'+\frac{1}{\sqrt{\eps}}\left(\frac{1}{\sqrt{\delta}}\left(\frac{1}{\sqrt{\eps'}}\left(\frac{1}{\sqrt{\delta'}}\mathsf{U}'+\mathsf{C}'\right)+\mathsf{T}\right)+\mathsf{C}\right) \\
&\textstyle\quad=s_1+s_2\sqrt{\frac{n}{s_1}}+s_1+\sqrt{\frac{n_2}{s_2}}\left(\sqrt{\frac{s_2}{m}}\left(\sqrt{\frac{nm}{s_1}}+m\right)+\sqrt{n}\right) \\
&\textstyle\quad= s_1+s_2\sqrt{\frac{n}{s_1}}+\frac{n}{\sqrt{s_1}}+\sqrt{n_2\frac{s_1^2}{n}}+\frac{n}{\sqrt{s_2}}
=\tilde{O}\left(s_1+s_2\sqrt{\frac{n}{s_1}}+\frac{n}{\sqrt{s_1}}+\frac{n}{\sqrt{s_2}}\right),
\end{align*}
which is still optimized by setting $s_1=n^{5/7}$ and $s_2=n^{4/7}$, 
giving time complexity $\tilde{O}(n^{5/7})$.

\section{Conclusion and Future Directions}

We have shown that the quantum walk search framework of \cite{MNRS11} can be extended to allow a data function that depends on both the vertex and the coin, provided certain costs are accounted for. This extension allows us to implement nested updates, although there may be other applications of this new framework, as it more generally allows us to consider updates with garbage resulting from any type of update subroutine.  Nested updates provide another tool for quantum walk algorithms analogous to the nested checking of \cite{JKM12}, and we hope that these tools will facilitate further upper bounds on both time and query complexity. 

It remains an open problem to improve the $\tilde{O}(n^{\frac{k}{k+1}})$ time complexity upper bound for $k$-Distinctness for $k>3$. The $3$-Distinctness upper bound of \cite{bel13} can be extended to a general $k$-Distinctness upper bound by coming up with an efficient procedure for constructing a starting state that generalizes our $\ket{\pi}_D^0$ \cite{bel13}. Efficiently constructing this generalized starting state would be a necessary, but not sufficient, condition for generalizing our upper bound to $k>3$. 

\section{Acknowledgments}

We thank Aleksandrs Belovs for helpful discussions about our two different approaches to this problem.

\bibliographystyle{alpha}
\bibliography{refs}
\end{document}